\documentclass[11pt]{amsart}
\usepackage{amscd}
\usepackage{verbatim}
\usepackage{amssymb}
\usepackage{amsmath}
\usepackage{amsthm}
\usepackage{amsfonts}
\usepackage{mathrsfs}
\usepackage{amsthm}
\usepackage{euscript}

\usepackage[margin=30 mm,heightrounded=true,centering]{geometry}

\newtheorem{theorem}{Theorem}[section]
\newtheorem{proposition}[theorem]{Proposition}

\newtheorem{remark}[theorem]{Remark}

\newtheorem{definition}[theorem]{Definition}

\begin{document}

\title[Scalar-tensor gravitation and the Bakry-\'Emery-Ricci tensor]{
Scalar-tensor gravitation and the Bakry-\'Emery-Ricci tensor}

\author{Eric Woolgar}
\address{Department of Mathematical and Statistical Sciences,
University of Alberta, Edmonton, Alberta, T6G 2G1, Canada}
\email{ewoolgar(at)ualberta.ca}

\date{\today}

\begin{abstract}
\noindent The Bakry-\'Emery generalized Ricci tensor arises in
scalar-tensor gravitation theories in the conformal gauge known as
the Jordan frame. Recent results from the mathematics literature
show that standard singularity and splitting theorems that hold when
an energy condition is applied in general relativity also hold when
that energy condition is applied to the Bakry-\'Emery tensor. We
show here that a direct consequence is that the Hawking-Penrose
singularity theorem and the timelike splitting theorem hold for
scalar-tensor theory in the Jordan frame. As examples, we consider
dilaton gravity (including totally anti-symmetric torsion) and the
Brans-Dicke family of scalar-tensor theories. For Brans-Dicke theory
the theorems do not extend to cover the entire space of values of
the Brans-Dicke family parameter $\omega$, and so may fail to hold
for $\omega<-1$. Observations show that this range of values does
not describe our Universe, but the result is in accord with examples
in the literature of Brans-Dicke spacetimes that have no singularity
in the Jordan frame and do not split as a Riemannian product.
\end{abstract}

\maketitle

\section{Introduction}
\setcounter{equation}{0}

\noindent Many theorems in general relativity hold under a
positivity assumption on components of the nongravitational
stress-energy tensor, usually called an energy condition, of which
there are various inequivalent forms. Because of the Einstein
equation, this assumption then becomes a condition on the Einstein
curvature and so, ultimately, on the Ricci curvature. Once one has a
condition on Ricci curvature, one has a tool to prove theorems.
Examples of these theorems include the Hawking-Penrose singularity
theorem (\cite{HE}, p 266), and the timelike splitting theorem
\cite{Eschenburg, Galloway1}, among others. The singularity theorem
is sometimes interpreted as the statement that nonnegative Ricci
curvature generically evolves to produce singularities, while the
splitting theorem shows that the nongeneric, nonsingular cases must
have quite special geometry.\footnote
{The term \emph{generic} in this sentence will actually acquire a
precise meaning below. Also, while splitting theorems endeavour to
describe the nongeneric cases, they typically do not fully describe
them, since some assumptions are needed; in particular, the
assumption of a complete timelike line; see Section 2.}

Scalar-tensor gravitation theories may be expressed in various
conformal gauges. The two most common such gauges are known as the
\emph{Einstein frame} and the \emph{Jordan frame}. In the Einstein
frame, it is usually a simple matter to show that standard theorems
will follow from an energy condition imposed on stress-energy. Note,
however, that this description of the theory is not ``minimally
coupled''. One may therefore prefer to impose the energy condition
on the theory in the Jordan frame, but in this frame an energy
condition imposed on nongravitational stress-energy does not
directly lead to a condition on the Ricci curvature from which one
could then derive these theorems (at least, in a direct manner by
following the usual general relativity proofs).

As well, other assumptions, which limit the solution space of a
scalar-tensor theory in the necessary manner to imply that these
theorems hold,\footnote
{An example of an assumption that is not conformally invariant is
the assumption of the existence of a complete timelike line in the
splitting theorem below.}
are not invariant in form under the conformal transformation between
``frames''. The net result is that, while singularity and splitting
theorems may hold for Einstein-frame solutions under certain
assumptions, similar theorems may not hold when assumptions of the
same form (thus, not conformally transformed versions of
Einstein-frame assumptions) are applied directly to Jordan-frame
solutions of the theory.

In this short note, I point out that recent progress in the
mathematics literature means that we now have available singularity
and splitting theorems arising from energy conditions and other
assumptions applied directly in the Jordan frame. The theorems
concerned are not obtained as corollaries from easily-established
Einstein frame theorems by simple methods. An interesting feature is
that, for Brans-Dicke theory, the theorems do not cover the case of
a Brans-Dicke parameter $\omega<-1$ (and require an additional
assumption in the special case of $\omega=-1$). This is consistent
with results reported in \cite{QBC}, which found that for
4-dimensional Brans-Dicke theory in the Jordan frame with a
barotropic perfect fluid and with Brans-Dicke parameter $\omega\le
-4/3$, there are nonsingular, nonsplit solutions.

To understand the mathematical progress that facilitates this, note
that in applications that touch upon Riemannian geometry, such as
optimal transportation \cite{Villani}, generalizations of the Ricci
curvature arise. An important example is the Bakry-\'Emery (or
Bakry-\'Emery-Ricci) tensor, which augments the Ricci tensor with
another term given by the Hessian of a twice-differentiable
\emph{weight function}. Some familiar theorems in Riemannian
geometry can be modified to hold under the assumption that the
Bakry-\'Emery tensor, rather than the usual Ricci tensor, obeys a
sign condition \cite{Lott}. In Lorentzian geometry, Case \cite{Case}
has shown that a similar sign condition on timelike components of
the Bakry-\'Emery tensor---i.e., an energy condition---will, in an
analogous fashion to the Riemannian case, imply that singularity
theorems and the timelike splitting theorem hold.

In Section 2 below, we state Case's versions of the singularity and
timelike splitting theorems and then explain some of the
terminology. In particular, we define the generalized Ricci tensor
and Bakry-\'Emery tensor and discuss the energy conditions. In
Section 3 we apply these theorems to obtain Jordan-frame singularity
and timelike splitting theorems for Brans-Dicke theory and (1-loop)
dilaton gravity, including dilaton gravity with totally skew
torsion. A brief concluding section contains some final remarks.

This research was supported by a Discovery Grant from the Natural
Sciences and Engineering Research Council of Canada. The author
thanks Greg Galloway for an explanation of the timelike splitting
theorem for globally hyperbolic spacetimes, and Valerio Faraoni for
comments on a draft and for bringing the reference \cite{BIT} to his
attention.

\section{Case's theorems}
\setcounter{equation}{0}

\noindent In \cite{Case}, Case has proved the following two
theorems:
\begin{theorem}[Case's Singularity Theorem]
\label{theorem2.1} Let $M$ be a chronological spacetime with $\dim
M\ge 3$. Let $\varphi:M\to {\mathbb R}$ obey the $\varphi$-generic
condition.\footnote
{Our usage differs from that of Case, who uses ``$f$-generic'' with
$f=-\log\varphi$. The $\varphi$-generic condition is that ${\dot
\gamma}_{[i}S_{j]pq[k}{\dot \gamma}_{l]}{\dot \gamma}^p {\dot
\gamma}^q$ is nonzero somewhere along each inextendible timelike
geodesic $\gamma$ (parametrized so that $g({\dot \gamma},{\dot
\gamma})=-1$), where $S^i{}_{jkl}:=R^i{}_{jkl}-\frac{1}{(n-1)}\left
( \frac{\nabla_j\nabla_l\varphi}{\varphi}
-\frac{n}{n-1}\frac{\nabla_j\varphi\nabla_l\varphi}{\varphi^2}\right
)\left ( \delta^i_k-{\dot \gamma}^i{\dot \gamma}_k\right )$. The
$\varphi$-generic condition implies that $S_{ijkl}{\dot
\gamma}^j{\dot \gamma}^l$ will be nonzero somewhere along $\gamma$;
see Definitions 2.7 and 3.1 of \cite{Case} (but in the present paper
we use the curvature tensor and terminology conventions of
\cite{HE}).}
Assume that either
\begin{enumerate}
\item there is an integer $q>0$ such that the generalized Ricci
    tensor ${\rm GRic}[g,\varphi,q]$ obeys the energy condition,
    or
\item the Bakry-\'Emery tensor ${\rm BER}[g,\varphi]$ obeys the
    energy condition and $\varphi$ is bounded away from zero\footnote
{Equivalently, $f=-\log\varphi$ is bounded above.}
($\varphi\ge C>0$ for some $C\in{\mathbb R}$).
\end{enumerate}
Assume also that either
\begin{enumerate}
\item[a)] $M$ has a point $p$ such that, along each null
    geodesic $\gamma$ through $p$, the modified null expansion
    scalar ${\hat \theta}:=\theta+\nabla_{\dot
    \gamma}\log\varphi$ of the null geodesic congruence through
    $p$ is negative somewhere to the future or past of $p$, or
\item[b)] $M$ has a closed $\varphi$-trapped surface,
    or\footnote
{Case uses the term ``$f$-trapped'' with $f=-\log \varphi$. We
define a closed (compact, without boundary), co-dimension 2,
$C^2$ surface $\Sigma$ to be \emph{$\varphi$-trapped} if the
modified null expansion scalars ${\hat \theta}$ of the two
oppositely directed null congruences leaving it orthogonally
obey ${\hat \theta}\vert_p\le 0$ at each $p\in\Sigma$. This
amounts to saying that each congruence is initially converging,
if convergence is measured with the rescaled metric ${\hat
g}=\varphi^{2/(n-2)}g$.}
\item[c)] $M$ has a compact spacelike hypersurface.
\end{enumerate}
Then $(M,g)$ is nonspacelike geodesically incomplete,
\end{theorem}

\begin{theorem}[Case's Timelike Splitting Theorem]
\label{theorem2.2} Let $(M,g)$ be a connected spacetime such that
\begin{enumerate}
\item $(M,g)$ is either timelike geodesically complete or
    globally hyperbolic,
\item $M$ contains a complete timelike line,\footnote
{Definition: a \emph{timelike line} $\gamma$ is an inextendible
timelike geodesic such that, for each pair of points $p,q$ along
$\gamma$ and every piecewise-smooth timelike curve joining these
points, the proper time interval along such curves from $p$ to
$q$ is maximized by $\gamma$. For the theorem, the line must be
\emph{complete}, meaning that its affine parameter takes values
throughout all of ${\mathbb R}$.}
and
\item either
\begin{enumerate}
\item there is an integer $q>0$ such that the generalized
    Ricci tensor ${\rm GRic}[g,\varphi,q]$ obeys the energy
    condition, or
\item the Bakry-\'Emery tensor ${\rm BER}[g,\varphi]$ obeys
    the energy condition and $\varphi$ is bounded away from zero.
\end{enumerate}
\end{enumerate}
Then $(M,g)$ is isometric to $({\mathbb R}\times \Sigma, -dt^2\oplus
h)$, where $(\Sigma,h)$ is a complete Riemannian manifold and
$\varphi$ is constant along ${\mathbb R}$.
\end{theorem}

We now undertake to explain some of the terminology used in these
theorems. First, for $(M,g)$ a Lorentzian $n$-manifold and
$\varphi$ a twice-differentiable function, we can define a family,
parametrized by $q$, of generalizations of the Ricci tensor (see,
e.g., \cite{Villani}). Taking $q\to\infty$ this yields the
Bakry-\'Emery tensor.
\begin{definition}\label{definition2.3}
Let $\varphi:M\to(0,\infty)$ be twice differentiable and let
$q\in(0,\infty)$. The \emph{generalized Ricci tensor} ${\rm
GRic}[g,\varphi,q]$, also denoted $G_{ij}$, is the tensor
\begin{eqnarray}
{\rm GRic}[g,\varphi,q]&:=&{\rm Ric}[g]-{\rm Hess}(\log \varphi)
-\frac1q \nabla \log \varphi \otimes \nabla \log \varphi\ ,\nonumber\\
&\equiv&{\rm Ric}[g]-\frac{{\rm Hess}(\varphi)}{\varphi}
+\left ( 1-\frac1q \right ) \frac{\nabla \varphi \otimes \nabla
\varphi}{\varphi^2} \ , \label{eq2.1}
\end{eqnarray}
where ${\rm Hess}(\varphi):=\nabla \nabla \varphi$ is the Hessian of
$\varphi$. The \emph{Bakry-\'Emery-Ricci tensor} ${\rm
BER}[g,\varphi]$, also denoted $B_{ij}$, is defined by
\begin{eqnarray}
{\rm BER}[g,\varphi]&:=&{\rm Ric}[g]-{\rm Hess}(\log \varphi)
\ ,\nonumber\\
&\equiv&{\rm Ric}[g]-\frac{{\rm Hess}(\varphi)}{\varphi}
+ \frac{\nabla \varphi \otimes \nabla \varphi}{\varphi^2}
\ . \label{eq2.2}
\end{eqnarray}
\end{definition}

The theorems above assume that either ${\rm BER}$ or ${\rm GRic}$
obeys what is called the energy condition. We now define energy
conditions. Our first such condition will apply to any
$(0,2)$-tensor:
\begin{definition}\label{definition2.4}
We say that a $(0,2)$-tensor $S$ obeys the \emph{energy condition}
on $(M,g)$ if, for every $p\in M$ and every timelike vector $t\in
T_pM$, then $S\big\vert_p(t,t)\ge 0$.
\end{definition}
\noindent By continuity, if a $(0,2)$-tensor $S$ obeys the energy
condition, then $S(l,l)\ge 0$ for all null vectors $l$ as well.

\noindent In general relativity, it is common to apply the energy
condition directly to the stress-energy tensor $T_{ij}$ of the
theory. This is then called the \emph{weak energy condition}, and in
general relativity it implies that the Einstein tensor obeys the
energy condition. If the energy condition is instead applied to
$T_{ij}-\frac{1}{n-2}g_{ij}T$, it is called the \emph{strong energy
condition}, and in general relativity this implies that the Ricci
tensor obeys the energy condition. With that in mind, assume now
that we are provided with a distinguished tensor $T_{ij}$ called the
stress-energy tensor, whether we are working in general relativity
or within a more general Lorentzian framework, such as that provided
another gravitation theory. We formulate an energy condition on
$T_{ij}$ that is natural in Brans-Dicke theory and which contains
the weak and strong energy conditions as special cases:
\begin{definition}\label{definition2.5}
We say the \emph{$\omega$-energy condition} holds for a given
$\omega>-(n-1)/(n-2)$ if
\begin{equation}
\left ( T_{ij}-\frac{(1+\omega)}{\left [ n-1+(n-2)\omega\right ] }
g_{ij}T\right )t^it^j\ge 0\label{eq2.4}
\end{equation}
for all timelike vectors $t^i\in T_pM$ and all $p\in M$. Taking
$\omega\to\infty$, we simply say that the \emph{strong energy
condition} holds if
\begin{equation}
\left ( T_{ij}-\frac{1}{n-2}g_{ij}T\right )t^it^j\ge 0 \label{eq2.5}
\end{equation}
for all timelike vectors $t^i\in T_pM$ and all $p\in M$. On the
other hand, we say the \emph{weak energy condition} holds if the
$\omega$-energy condition holds for $\omega=-1$; that is, if
\begin{equation}
T_{ij}t^it^j\ge 0 \label{eq2.3}
\end{equation}
for all timelike vectors $t^i\in T_pM$ and all $p\in M$.
\end{definition}
\noindent That is, the weak energy condition holds if the energy
condition holds for $S_{ij}=T_{ij}$, and the strong energy condition
holds when the energy conditions holds for
$S_{ij}=T_{ij}-\frac{1}{n-2}g_{ij}T$. The terms \emph{weak, strong,}
and \emph{$\omega$-energy condition} will always refer to $T_{ij}$,
whereas we have used \emph{energy condition} to refer to any tensor.

The $\omega$-energy condition holds for all $\omega\ge -1$ if and
only if matter obeys both the weak and the strong energy condition.

The $\omega$-energy condition for any $\omega$ reduces to the weak
energy condition if matter consists only of massless radiation,
since then $T=0$.

\section{Applications to scalar-tensor gravitation}
\setcounter{equation}{0}

\subsection{Dilaton gravity}

\noindent It is believed that low energy string theory is described
by a nonlinear sigma model at a fixed point of its renormalization
group flow. The precise model depends on whether torsion in the
guise of the so-called ${\mathcal B}$-field is present,\footnote
{Most authors use $B$ for this field, but we use
${\mathcal B}$ to distinguish it from the Bakry-\'Emery tensor $B_{ij}$.}
and whether the string theory in question is bosonic,
supersymmetric, or heterotic. The fixed point condition means that
the so-called beta-functions of the theory vanish. Of particular
interest will be the condition for the vanishing of the graviton
beta-function, $\beta_g$. This is a condition on the metric of the
target manifold and can be expressed as
\begin{equation}
0=\beta_g\equiv \alpha' \left (R_{ij}+2\nabla_i\nabla_j \Phi
- \frac14 {\mathcal H}_{ikl}{\mathcal H}_j{}^{kl}\right )
+{\mathcal O}(\alpha'^2) \ . \label{eq3.1}
\end{equation}
Here $\alpha'$ is a constant, $\Phi$ is a scalar field called the
{\it dilaton field}, and the 3-form ${\mathcal H}$ is the field
strength tensor ${\mathcal H}=d{\mathcal B}$ of a 2-form field
${\mathcal B}$. The ${\mathcal O}(\alpha'^2)$ terms are in fact a
power series in $\alpha'$, and for the bosonic nonlinear sigma model
the terms consist of terms of quadratic and higher order in the
Riemann tensor, $\nabla\Phi$, ${\mathcal H}$, and their derivatives
(for other models, additional fields can eventually appear). These
terms are exactly determined, but are known only to rather low order
in $\alpha'$ (the precise order depends on the sigma model).
Therefore, we will write (\ref{eq3.1}) in the form
\begin{equation}
R_{ij}+2\nabla_i\nabla_j \Phi
=\frac14 {\mathcal H}_{ikl}{\mathcal H}_j{}^{kl}
+8\pi\alpha' \tau_{ij}\ , \label{eq3.2}
\end{equation}
where $\tau_{ij}$ is a power series in $\alpha'$ that would be
completely determined if we had the ability to compute $\beta_g$ to
all orders in $\alpha'$. Instead of $\tau_{ij}$, to make contact
with general relativity we will use the corresponding
``stress-energy tensor'' $T_{ij}$ defined by
$T_{ij}:=\tau_{ij}-\frac12 g_{ij} \tau$ with $\tau:=g^{ij}T_{ij}$.
We also prefer to replace the dilaton by
\begin{equation}
\varphi=e^{-2\Phi}\ . \label{eq3.3}
\end{equation}
The left-hand side of (\ref{eq3.2}) becomes
\begin{equation}
R_{ij}+2\nabla_i\nabla_j\Phi=R_{ij}-\frac{1}{\varphi}\nabla_i\nabla_j
\varphi +\frac{1}{\varphi^2}\nabla_i\varphi\nabla_j\varphi
\equiv B_{ij}[g,\varphi]\ , \label{eq3.4}
\end{equation}
and so the condition for the vanishing of $\beta_g$ becomes
\begin{equation}
B_{ij}[g,\varphi]
=\frac14 {\mathcal H}_{ikl}{\mathcal H}_j{}^{kl}+8\pi\tau_{ij}
\equiv \frac14 {\mathcal H}_{ikl}{\mathcal H}_j{}^{kl}
+8\pi\alpha' \left ( T_{ij}
-\frac{1}{n-2}g_{ij}T\right )\ . \label{eq3.5}
\end{equation}

\begin{proposition}\label{proposition3.1}
Let $(M,g)$ be a chronological spacetime with $\dim M\ge 3$. Say
that $T_{ij}$ obeys the strong energy condition, that ${\mathcal
H}_{ijk}$ is zero or totally skew, that (\ref{eq3.5}) holds, that
there is a $C\in {\mathbb R}$, $C>0$, such that $\varphi(p)\ge C$
for all $p\in M$, and that the $\varphi$-generic condition holds.
Assume further that at least one of the conditions (a), (b), or (c)
from Theorem \ref{theorem2.1} hold. Then $(M,g)$ is nonspacelike
geodesically incomplete.
\end{proposition}

\begin{proof}
Choose any timelike vector $t^i$ and construct an
orthonormal frame $\{e_0,e_{\alpha}:\alpha=1,\dots ,n-1\}$ aligned
with $t^i$, so $t^i=|t|e_0$ where $|t|:=\sqrt{-g_{ij}t^it^j}$. For
totally skew torsion, we then have ${\mathcal H}_i{}^{kl}{\mathcal
H}_{jkl}t^it^j=+|t|^2({\mathcal H}_{0\alpha\beta})^2$ since total
skewness implies that neither of the ${\mathcal H}$ factors can have
more than one $0$-index. Using this and the strong energy condition,
we see that the right-hand side of (\ref{eq3.5}) is nonnegative
whenever it is contracted against $t^it^j$ for any timelike vector
$t^i$. Then the result is an immediate consequence of Theorem
\ref{theorem2.1}.
\end{proof}

\begin{proposition}\label{proposition3.2}
Let $\dim M\ge 3$. Say that $T_{ij}$ obeys the strong energy
condition, that ${\mathcal H}_{ijk}$ is zero or totally skew, that
${\rm BER}$ is given by (\ref{eq3.5}), and that there is a $C\in
{\mathbb R}$, $C>0$, such that $\varphi(p)\ge C$ for all $p\in M$.
If $(M,g)$ is either globally hyperbolic or timelike geodesically
complete, and admits a complete timelike line, then $(M,g)\simeq
({\mathbb R}\times {\hat M}, -dt^2+{\hat g})$, so $(M,g)$ splits off
the timelike line and $\varphi$ is constant.
\end{proposition}

\begin{proof}
This follows from the same calculation as in the previous proof, but
now we apply Theorem \ref{theorem2.2} instead of Theorem
\ref{theorem2.1}.

\end{proof}

\subsection{Brans-Dicke theory in the Jordan frame}

\noindent The Brans-Dicke theory \cite{BD} on a manifold $M$ is
actually a family of theories, parametrized by the Brans-Dicke
parameter $\omega> -\left ( \frac{n-1}{n-2}\right )$. There are two
gravitational fields, a Lorentzian metric tensor $g_{ij}$ and a
scalar field $\varphi$ which is taken to be everywhere positive.
These fields obey the system of equations (\cite{Will}, p 123)
\begin{equation}
R_{ij}-\frac12 g_{ij}R=\frac{8\pi}{\varphi}T_{ij}
+\frac{\omega}{\varphi^2}\left ( \nabla_i\varphi\nabla_j\varphi
-\frac12 g_{ij}g^{kl}\nabla_k\varphi\nabla_l\varphi \right )
+\frac{1}{\varphi} \left ( \nabla_i\nabla_j \varphi
-g_{ij}\square \varphi \right ) \ , \label{eq3.6}
\end{equation}
\begin{equation}
\square \varphi -\frac{1}{2\varphi}g^{kl}\nabla_k\varphi\nabla_l\varphi
+\frac{1}{2\omega} R \varphi =0\ , \label{eq3.7}
\end{equation}
where $\square\varphi:=\frac{1}{\sqrt{-g}}\partial_i \left (
\sqrt{-g} g^{ij}\partial_j \phi \right
)=g^{ij}\nabla_i\nabla_j\varphi$ is the d'Alembertian of $\varphi$.

Conformal transformations will change the form of these equations.
The conformal choice that leads to the above form is usually called
the {\it Jordan frame} in the literature. Note that notions like
\emph{closed trapped surface} and \emph{timelike line} are
\emph{not} conformally invariant.

Equation (\ref{eq3.6}) can be rewritten as
\begin{equation}
R_{ij}-\frac{1}{\varphi}\nabla_i\nabla_j\varphi = \frac{8\pi}{\varphi}
\left ( T_{ij}-\frac{1}{(n-2)}g_{ij}T\right ) +\frac{\omega}{\varphi^2}
\nabla_i\varphi\nabla_j\varphi +\frac{1}{(n-2)}\frac{\square\varphi}{\varphi}
g_{ij}\ . \label{eq3.8}
\end{equation}
We can also rewrite (\ref{eq3.7}) as
\begin{equation}
\square\varphi =\frac{8\pi T}{\left [ n-1 +(n-2)\omega \right ]}
\ . \label{eq3.9}
\end{equation}
Inserting this in (\ref{eq3.8}) yields
\begin{equation}
R_{ij}-\frac{1}{\varphi}\nabla_i\nabla_j\varphi
-\frac{\omega}{\varphi^2} \nabla_i\varphi\nabla_j\varphi
= \frac{8\pi}{\varphi}\left ( T_{ij}-\frac{(1+\omega)}{\left [
n-1+(n-2)\omega\right ] }g_{ij}T\right )\ . \label{eq3.10}
\end{equation}

\begin{proposition}\label{proposition3.3}
Let $(M,g)$ be a chronological spacetime with $\dim M\ge 3$. Say
that either (i) for some fixed $\omega> -1$, the pair $(g,\varphi)$
obeys the system (\ref{eq3.6}, \ref{eq3.7}), where $T_{ij}$ obeys
the $\omega$-energy condition, or (ii) for $\omega=-1$ the pair
$(g,\varphi)$ obeys the system (\ref{eq3.6}, \ref{eq3.7}) where
$T_{ij}$ obeys the weak energy condition and $\varphi\ge C >0$ for
some constant $C$. Assume further that the $\varphi$-generic
condition holds and that at least one of the conditions (a), (b), or
(c) from Theorem \ref{theorem2.1} hold. Then $(M,g)$ is nonspacelike
geodesically incomplete.
\end{proposition}

\begin{proof}
Equations (\ref{eq3.6}, \ref{eq3.7}) imply (\ref{eq3.10}), which in
turn can be written as
\begin{equation}
G_{ij}\left [g,\varphi,\frac{1}q\right ]= \frac{8\pi}{\varphi}
\left ( T_{ij}-\frac{(1+\omega)}{\left [ n-1+(n-2)\omega\right ] }
g_{ij}T\right ) +\left (1+\omega-\frac1q \right )\frac{\nabla_i\varphi
\nabla_j\varphi}{\varphi^2}\ . \label{eq3.11}
\end{equation}
If $\omega>-1$, choose a positive integer $q\ge \frac{1}{1+\omega}$.
For such a $q$, and using the $\omega$-energy condition, the
right-hand side of (\ref{eq3.11}) is $\ge 0$ when contracted with
$t^it^j$ for any timelike vector $t^j$. Thus, $G_{ij}\left
[g,\varphi,\frac{1}q\right ]t^it^j\ge 0$. The result then follows
from Theorem \ref{theorem2.1}.

If instead $\omega=-1$, then taking $q\to \infty$ in (\ref{eq3.11})
and again invoking the $\omega$-energy condition, the right-hand
side of (\ref{eq3.11}) is again $\ge 0$ when contracted with
$t^it^j$ for any timelike vector $t^j$, implying now that
$B_{ij}[g,\varphi]t^it^j\ge 0$. Using the boundedness of $\varphi$,
the result again follows from Theorem \ref{theorem2.1}.
\end{proof}

\begin{proposition}\label{proposition3.4}
Let $\dim M\ge 3$. Say that either (i) for some fixed $\omega> -1$,
the pair $(g,\varphi)$ obeys the system (\ref{eq3.6}, \ref{eq3.7}),
where $T_{ij}$ obeys the $\omega$-energy condition, or (ii) for
$\omega=-1$ the pair $(g,\varphi)$ obeys the system (\ref{eq3.6},
\ref{eq3.7}) where $T_{ij}$ obeys the weak energy condition and
$\varphi\ge C>0$. If $(M,g)$ is either globally hyperbolic
or timelike geodesically complete, and admits a complete timelike
line, then $(M,g)\simeq ({\mathbb R}\times {\hat M}, -dt^2+{\hat
g})$ and $\varphi$ is constant.
\end{proposition}
\noindent Thus $(M,g)$ is a static solution of general relativity.
\begin{proof}
Follow the same argument as in the proof of the previous
proposition, but instead of Theorem \ref{theorem2.1}, invoke Theorem
\ref{theorem2.2}.
\end{proof}

In the case of globally hyperbolic spacetimes, a nice example of the
necessity of the assumption that the timelike line in the splitting
theorem must be complete is provided by the O'Hanlon-Tupper family
\cite{OHT}, each member of which has incomplete timelike lines and
otherwise satisfies all assumptions of Proposition
\ref{proposition3.4}, and which only splits as a warped product with
nonconstant $\varphi$ (which is monotonic along the timelike line
and not bounded away from zero).

\section{Concluding remark}
\setcounter{equation}{0}

\noindent The theorems do not cover the case of $\omega<-1$.
Consider $n=4$ dimensions. Of course, this is the physical case, and
observations currently imply that Brans-Dicke theory cannot describe
gravity in the solar system unless $\omega>4\times 10^4$ \cite{BIT}.
Nevertheless, we can consider solutions with $3/2<\omega<-1$. As
mentioned in the introduction, there are known $n=4$ nonsingular
solutions for $-3/2<\omega\le -4/3$ which obey the assumptions of
the splitting theorem, including the $\omega$-energy
condition\footnote
{Solutions are found in \cite{QBC} for values of a parameter
$\gamma$, the barotropic index of the perfect fluid, in the range
$0<\gamma<2$. For $\omega<-1$ and $n=4$ dimensions, the
$\omega$-energy condition holds when $\gamma\le 1$. For $\omega\le
-4/3$, as in the \cite{QBC} solutions, the $\omega$-energy condition
holds for $\gamma\le 5/4$.}
and the existence of a complete timelike line, but do not split
\cite{QBC} (they do split as a warped product, but not as a
product). Thus, the splitting theorem cannot be extended to all
allowed $\omega$. Whether there can be $n=4$ nonsplit, nonsingular
solutions of Brans-Dicke theory with $-4/3<\omega<-1$ is an open
question.

Finally, we return to the issue of whether the Brans-Dicke results
are genuinely unrelated to the Einstein frame. With the splitting
theorem, there is no apparent reason to question that it is an
entirely independent result. For the singularity theorem, this may
be less clear. On the one hand, by working entirely in the Jordan
frame, we see that $\omega=-1$ is singled out as a boundary case,
previous results \cite{QBC} having already indicated that there
should be a boundary case. As well, the form of the theorem that
employs assumption (c) of Theorem \ref{theorem2.1} seems
satisfactory. However, the form of the theorem that employs
assumptions (a) and (b) seem rather less satisfactory, since they
both are phrased as a condition on ${\hat \theta}$, which is related
to the null expansion $\theta$ by the same conformal transformation
that relates the Einstein and Jordan frames. Hence, these
assumptions still echo an Einstein frame formulation. Perhaps this
is the best that can be done, since it may be necessary to modify
the mean curvature of a trapped surface to overcome a defocusing
effect due to the interaction of curvature with the Brans-Dicke
scalar field.\footnote
{This can work both ways: One can have $\theta>0$ and yet ${\hat
\theta}\le 0$, trapping surfaces that would be, in general
relativity, not trapped; then the scalar field would contribute an
additional focusing, rather than defocusing, effect.}
For now, whether these assumptions can be modified to refer only to
$\theta$ remains an open question.

\end{document}